\documentclass[prl,twocolumn,superscriptaddress,nofootinbib,showpacs]{revtex4-1}
\usepackage{bm} % for bold math font
\usepackage{amsmath} % AMS math
\usepackage{amssymb} % AMS symbols
\usepackage{amsthm} % AMS theorem environment

\def\boldpi{\boldsymbol{\pi}}
\def\boldsigma{\boldsymbol{\sigma}}

\def\sdotp{\boldsigma\cdot\boldpi}
\def\sdotE{\boldsigma\cdot\mathbf{E}}
\def\sdotB{\boldsigma\cdot\mathbf{B}}
\def\Edotp{\mathbf{E}\cdot\boldpi}
\def\Bdotp{\mathbf{B}\cdot\boldpi}
\def\Etimesp{\mathbf{E}\times\boldpi}

\def\calH{\mathcal{H}}
\def\calX{\mathcal{X}}
\def\calY{\mathcal{Y}}
\def\calZ{\mathcal{Z}}

\def\mbeta{\tilde{\beta}}
\def\malpha{\tilde{\alpha}}
\def\mbalpha{\tilde{\boldsymbol{\alpha}}}

\def\mbsigma{\tilde{\boldsymbol{\sigma}}}
\def\mgamma{\tilde{\gamma}}

\def\bigO{O}

\newcommand{\eqnref}[1]{(\ref{#1})}

\theoremstyle{definition}
\newtheorem{theorem}{Theorem}

\begin{document}

\title
{Exact correspondence between classical and Dirac-Pauli spinors in the weak-field limit of static and homogeneous electromagnetic fields}

\thanks{This article serves as a brief summary of the main result of \cite{Chen2014b}.}

\author{Dah-Wei Chiou}
\email{dwchiou@gmail.com}
\affiliation{Department of Physics, National Taiwan Normal University, Taipei 11677, Taiwan}
\affiliation{Center for Condensed Matter Sciences, National Taiwan University, Taipei 10617, Taiwan}

\author{Tsung-Wei Chen}
\email{twchen@mail.nsysu.edu.tw} \affiliation{Department of Physics, National Sun Yat-sen University, Kaohsiung 80424, Taiwan}

%\date{\today}

\begin{abstract}
It has long been speculated that the Dirac or, more generally, the Dirac-Pauli spinor in the Foldy-Wouthuysen (FW) representation should behave like a classical relativistic spinor in the low-energy limit when the probability of particle-antiparticle pair creation and annihilation is negligible. In the weak-field limit of static and homogeneous electromagnetic fields, by applying the method of direct perturbation theory inductively on the orders of $1/c$ in the power series, we rigorously prove that it is indeed the case: the FW transformation of the Dirac-Pauli Hamiltonian is in full agreement with the classical counterpart, which is the sum of the orbital Hamiltonian for the Lorentz force equation and the spin Hamiltonian for the Thomas-Bargmann-Michel-Telegdi equation.
\end{abstract}


\maketitle

The relativistic quantum theory for a spin-$1/2$ particle subject to external electromagnetic fields is described by the Dirac equation  \cite{Dirac1928,Dirac1982}
\begin{equation*}\label{Dirac eq}
\mgamma^\mu D_\mu|\psi\rangle+i\frac{mc}{\hbar}|\psi\rangle=0,
\end{equation*}
where the covariant derivative $D_\mu$ is given by
$D_\mu:=\partial_\mu +\frac{iq}{\hbar c}A_\mu \equiv -\frac{i}{\hbar}\pi_\mu :=-\frac{i}{\hbar}\left(p_\mu-\frac{q}{c}A_\mu\right)$
with the canonical 4-momentum $p^\mu=(E/c,\mathbf{p})$ and the kinematic 4-momentum $\pi^\mu=(W/c,\boldsymbol{\pi})$, and $\mgamma^\mu$ are the $4\times4$ Dirac matrices.
The Dirac equation gives rise to the magnetic moment with the $g$-factor given by $g=2$. To incorporate any anomalous magnetic moment (i.e.\ $g\neq2$), one can augment the Dirac equation into the Dirac-Pauli equation \cite{Pauli1941} with an additional term explicitly dependent on the field strength $F_{\mu\nu}:=\partial_\mu A_\nu-\partial_\nu A_\mu$:
\begin{equation*}\label{Dirac-Pauli eq}
\mgamma^\mu D_\mu|\psi\rangle+i\frac{mc}{\hbar}|\psi\rangle
+\frac{i\mu'}{2c}\mgamma^\mu\mgamma^\nu F_{\mu\nu}|\psi\rangle=0,
\end{equation*}
which amounts to including the anomalous magnetic moment given by $\mu'=\gamma'_m\hbar/2$ (where $\gamma'_m$ is called the anomalous gyromagnetic ratio).
The Pauli-Dirac equation can be cast in the Hamiltonian formalism as
\begin{equation*}\label{Dirac-Pauli eq w H}
i\hbar\frac{\partial}{\partial t}|\psi\rangle = \tilde{\calH}|\psi\rangle
\end{equation*}
with the Dirac Hamiltonian $\tilde{H}$ and the Dirac-Pauli Hamiltonian $\tilde{\calH}$ read as
%\begin{subequations}
\begin{eqnarray*}\label{H Dirac and Dirac-Pauli}
\label{H Dirac}
\tilde{H} &=& mc^2\mbeta+c\,\mbalpha\cdot\left(\mathbf{p}-\frac{q}{c}\mathbf{A}\right)+q\phi, \\
%&\equiv&
%\left(
%  \begin{array}{cc}
%    mc^2+q\phi & c\,\boldsigma\cdot\boldpi \\
%    c\,\boldsigma\cdot\boldpi & -mc^2+q\phi \\
%  \end{array}
%\right),\nonumber\\
\label{H Dirac-Pauli}
\tilde{\calH} &=& \tilde{H} +\mu'\left(-\mbeta\mbsigma\cdot\mathbf{B}+i\mbeta\mbalpha\cdot\mathbf{E}\right), %\\
%&\equiv&
%\left(
%  \begin{array}{cc}
%    mc^2+q\phi-\mu'\sdotB & c\,\boldsigma\cdot\boldpi+i\mu'\sdotE \\
%    c\,\boldsigma\cdot\boldpi-i\mu'\sdotE & -mc^2+q\phi+\mu'\sdotB \\
%  \end{array}
%\right),\nonumber
\end{eqnarray*}
%\end{subequations}
where the $4\times4$ matrices are given explicitly by
\begin{equation*}
\mbeta=\left(
         \begin{array}{cc}
           \openone & 0 \\
           0 & -\openone \\
         \end{array}
       \right),\
\mbalpha=\left(
           \begin{array}{cc}
             0 & \boldsigma \\
             \boldsigma & 0 \\
           \end{array}
         \right),\
\mbsigma=\left(
           \begin{array}{cc}
             \boldsigma & 0 \\
             0 & \boldsigma \\
           \end{array}
         \right),
\end{equation*}
and $\boldsigma=(\sigma_x,\sigma_y,\sigma_z)$ are the $2\times2$ Pauli matrices.
Accordingly, the $\mgamma$ matrices are given by $\mgamma^0=\mbeta$ and $\mgamma^i=\mbeta\malpha^i$.

Rigorously, the Dirac equation is self-consistent only in the context of quantum field theory as particle-antiparticle pairs can be created and annihilated. (And accordingly, the Dirac-Pauli equation accounting for the anomalous magnetic moment is adequate only at the phenomenological level.) In the weak-field limit when the particle's energy interacting with electromagnetic fields is much smaller than the Dirac energy gap $2mc^2$, the probability of pair creation and annihilation is negligible and it is expected that the particle and antiparticle can be treated separately without taking into account the field-theory interaction between them.
The Foldy-Wouthuysen (FW) transformation is the method for the particle-antiparticle decomposition via a series of successive unitary transformations, each of which block-diagonalizes the Dirac Hamiltonian to a certain order of $1/m$ \cite{Foldy1950} (also see \cite{Strange2008} for a review). In the same spirit, many different approaches have been developed for various advantages and most of them can be straightforwardly applied to the Dirac-Pauli Hamiltonian (see \cite{Silenko1995-analysis,Reiher2009 Book} for reviews).

On the other hand, the classical (non-quantum) dynamics for a relativistic point particle endowed with charge and intrinsic spin in electromagnetic fields is well understood. The orbital motion, which is governed by the Lorentz force equation, and the precession of spin, which is govern by the Thomas-Bargmann-Michel-Telegdi equation \cite{Thomas1927,Bargmann1959} (see \cite{Jackson1999} for a review), are simultaneously described by the total Hamiltonian (see \cite{Chen2014} for more details)
\begin{eqnarray}\label{H cl}
&&H(\mathbf{x},\mathbf{p},\mathbf{s};t)\\
&=&H_\mathrm{orbit}(\mathbf{x},\mathbf{p};t)
+H_\mathrm{spin}(\mathbf{s},\mathbf{x},\mathbf{p};t)+\bigO(F_{\mu\nu}^2,\hbar^2)\nonumber
\end{eqnarray}
with the orbital Hamiltonian given by
\begin{equation}\label{H orbit}
H_\mathrm{orbit}(\mathbf{x},\mathbf{p};t)=\sqrt{m^2c^4+c^2\boldpi^2}+q\phi(\mathbf{x},t)
\end{equation}
and the spin Hamiltonian given by
\begin{widetext}
\begin{eqnarray}\label{H spin}
&&H_\mathrm{spin}(\mathbf{s},\mathbf{x},\mathbf{p};t)\\
&=&-\mathbf{s}\cdot\left[
\left(\gamma'_m+\frac{q}{mc}\frac{1}{\gamma_{\boldpi}}\right)\mathbf{B}(\mathbf{x})
\right.%\nonumber\\
%&&\qquad
\mbox{}-\gamma'_m\frac{1}{\gamma_{\boldpi}(1+\gamma_{\boldpi})}
\left(\frac{\boldpi}{mc}\cdot\mathbf{B}(\mathbf{x})\right)\frac{\boldpi}{mc}%\nonumber\\
%&&\qquad
\left.
\mbox{}-\left(\gamma'_m\frac{1}{\gamma_{\boldpi}} +\frac{q}{mc}\frac{1}{\gamma_{\boldpi}(1+\gamma_{\boldpi})}\right)
\left(\frac{\boldpi}{mc}\times\mathbf{E}(\mathbf{x})\right)
\right],\nonumber
\end{eqnarray}
\end{widetext}
where $\mathbf{s}$ is the intrinsic spin and the Lorentz factor associated with the kinematic momentum $\boldsymbol{\pi}$ is defined as $\gamma_{\boldsymbol{\pi}}:=\sqrt{1+(\boldsymbol{\pi}/(mc))^2}$.
The Hamiltonian $H(\mathbf{x},\mathbf{p},\mathbf{s};t)$ provides a low-energy description of the relativistic spinor dynamics.

It has long been conjectured that, in the weak-field limit, the Dirac or Dirac-Pauli Hamiltonian, after block diagonalization, should agree with the classical hamiltonian $H(\mathbf{x},\mathbf{p},\mathbf{s};t)$ up to corrections of $\bigO(F_{\mu\nu}^2,\hbar^2)$ (except that the spin of the Dirac-Pauli spinor is quantized). The classical-quantum correspondence has been suggested and investigated from different aspects with various degrees of rigor \cite{Silenko2008,Chen2014,Chen2010,Chen2013,Rubinow1963,Rafa1964,Froh1993,Silenko1995}.

In this article, we consider the case subject to static and homogeneous fields, whereby the $\bigO(\hbar^2)$ corrections arising from the operator ordering and the Darwin term are absent and the FW transformation remains explicitly time-independent and thus in conformity with the standard FW scenario \cite{Chen2014}. Furthermore, we neglect all nonlinear electromagnetic corrections of $\bigO(F_{\mu\nu}^2)$ in the weak-field limit. In these settings, by mathematical induction on the orders of $1/c$ in the power series, we rigourously prove that the conjectured classical-quantum correspondence is \emph{exact}, first for the Dirac equation and then for the Dirac-Pauli equation. (More details and other related issues are presented in a separated paper \cite{Chen2014b}.)

We adopt the method of direct perturbation theory (DPT) \cite{Rutkowski1986a,Rutkowski1986b,Rutkowski1986c,Heully1986}, in the style of Kutzelnigg's implementation \cite{Kutzelnigg1990} with a further simplification scheme introduced in \cite{Chen2013}, to obtain the FW transformation. In Kutzelnigg's implementation of DPT, the FW unitary transformation is assumed to take the form
\begin{equation*}\label{U and Udag}
\tilde{U}=
\left(
  \begin{array}{cc}
    \calY & \calY\calX^\dag \\
    -\calZ\calX & \calZ \\
  \end{array}
\right),
\qquad
\tilde{U}^\dag=
\left(
  \begin{array}{cc}
    \calY & -\calX^\dag\calZ \\
    \calX\calY & \calZ \\
  \end{array}
\right),
\end{equation*}
where the $2\times2$ hermitian operators $\calY$ and $\calZ$ are
\begin{equation*}
\calY=\calY^\dag=\frac{1}{\sqrt{1+\calX^\dag\calX}}, \qquad
\calZ=\calZ^\dag=\frac{1}{\sqrt{1+\calX\calX^\dag}}
\end{equation*}
for some operator $\calX$ to be determined.
The FW transformed Hamiltonian is given by
\begin{equation*}
\tilde{\calH}_\mathrm{FW}
\equiv
\left(
  \begin{array}{cc}
    \calH_\mathrm{FW} & 0 \\
    0 & \bar{\calH}_\mathrm{FW} \\
  \end{array}
\right)
=\tilde{U}\tilde{\calH}\tilde{U}^\dag,
\end{equation*}
where $\bar{\calH}_\mathrm{FW}(\mathbf{x},\boldpi,\boldsigma;q,\mu')
=-\calH_\mathrm{FW}(\mathbf{x},-\boldpi,-\boldsigma;-q,-\mu')$ by $CPT$ symmetries.
For the Dirac-Pauli theory, the block-diagonality of $\tilde{\calH}_\mathrm{FW}$ entails the constraint upon $\calX$ as
\begin{eqnarray}\label{condition calX}
{2mc^2}\calX&=&-\calX c\,\sdotp\,\calX+c\,\sdotp+q[\phi,\calX]\\
&& \mbox{}-i\mu'\sdotE-i\mu'\calX\sdotE\,\calX +\mu'\{\calX,\sdotB\}, \nonumber
\end{eqnarray}
and correspondingly the FW transformed Hamiltonian is given by
%\begin{widetext}
\begin{eqnarray}\label{calHFW Kutzelnigg}
\calH_\mathrm{FW}
&=& mc^2 + \sqrt{1+\calX^\dag\calX}\,
\Big(q\phi +c\,\sdotp\calX \nonumber\\
&&\quad\mbox{}-\mu'\sdotB +i\mu'\sdotE\,\calX\Big)
\frac{1}{\sqrt{1+\calX^\dag\calX}}.
\end{eqnarray}
%\end{widetext}
Particularly, for the Dirac theory, \eqnref{condition calX} and \eqnref{calHFW Kutzelnigg} reduce to (by simply setting $\mu'=0$)
\begin{equation}\label{condition X}
2mc^2X=
-Xc\,\sdotp X+c\,\sdotp+q[\phi,X],
\end{equation}
and
\begin{eqnarray}\label{HFW Kutzelnigg}
&& H_\mathrm{FW} \\
&=& mc^2 + \sqrt{1+X^\dag X}\,\left(q\phi+c\,\sdotp X\right)\frac{1}{\sqrt{1+X^\dag X}}, \nonumber
\end{eqnarray}
where we have used the notations $X$ and $H_\mathrm{FW}$ in place of $\calX$ and $\calH_\mathrm{FW}$ when the Dirac-Pauli theory is reduced to the Dirac theory.

First, let us consider the Dirac theory. By expanding $X$ in powers of $c^{-1}$ as
\begin{equation*}
X=\sum_{j=1}^\infty\frac{X_j}{c^j},
\end{equation*}
\eqnref{condition X} yields
\begin{equation}\label{X1 X2}
2mX_1=\sdotp, \qquad 2mX_2=0,
\end{equation}
and the recursion relations (for $j\geq1$):
\begin{subequations}\label{recursion Xj}
\begin{eqnarray}
\label{recursion a}
2mX_{2j}=-\sum_{k_1+k_2=2j-1}\!\!\!X_{k_1}\sdotp X_{k_2}+q[\phi,X_{2j-2}],\qquad\\
\label{recursion b}
2mX_{2j+1}=-\sum_{k_1+k_2=2j}\!\!\!X_{k_1}\sdotp X_{k_2}+q[\phi,X_{2j-1}].\qquad\;
\end{eqnarray}
\end{subequations}
These allow one to compute $X_n$ to any desired order.

Thanks to the high-order calculation conducted in \cite{Chen2013} up to $X_{14}$, we can
conjecture the generic expression of $X_n$ and have the following theorem:
\begin{theorem}\label{thm:1}
In the weak-field limit, we neglect nonlinear terms in $\mathbf{E}$ and $\mathbf{B}$. If the electromagnetic field is homogeneous (thus, $[\pi_i,E_j]=[\pi_i,B_j]=0$), the generic expression for $X_{n\geq0}$ is given by
\begin{subequations}\label{X2j and X2j+1}
\begin{eqnarray}
\label{X2j}
X_{2j}&=&0,\\
\label{X2j+1}
X_{2j+1}&=&a_j\frac{(-1)^j}{(2m)^{2j+1}}(\sdotp)^{2j+1}\nonumber\\
&&\mbox{}
+b_j\frac{iq\hbar(-1)^j}{(2m)^{2j}}\,\boldpi^{2j-2}(\sdotE)\nonumber\\
&&\mbox{}
+c_j\frac{iq\hbar(-1)^j}{(2m)^{2j}}\,\boldpi^{2j-4}(\sdotp)(\Edotp),
\end{eqnarray}
\end{subequations}
where the coefficients are defined as
\begin{subequations}\label{aj bj cj}
\begin{eqnarray}
\label{aj}
a_{j\geq0} &=& \frac{(2j)!}{j!(j+1)!},\\
\label{bj}
b_{j\ge1} &=& \frac{(2j-1)!}{j!(j-1)!} \equiv (2j-1)a_{j-1}, \quad b_{j=0}=0,\quad\\
\label{cj}
c_{j\geq0}&=&2\sum_{j_1+j_2=j}b_{j_1}b_{j_2}, \quad (\text{note},\ c_{j=0,1}=0).
\end{eqnarray}
\end{subequations}
\end{theorem}
\begin{proof}
It is trivial to prove \eqnref{X2j} by applying \eqnref{recursion a} on \eqnref{X1 X2} inductively. After knowing $X_{2j}=0$, \eqnref{X2j+1} is proven by mathematical induction via \eqnref{recursion b} with the help of \eqnref{identity 4}, and \eqnref{combo identity a}--\eqnref{combo identity c}.
\end{proof}

Once $X_n$ are known, we can express $X^\dag X$ and $XX^\dag$ in the form of power series. Neglecting nonlinear terms in $F_{\mu\nu}$, we have $\left[c\,\sdotp X, X^\dag X\right] =0$
and, by induction, $[q\phi,(X^\dag X)^n] =n[q\phi,X^\dag X](X^\dag X)^{n-1}$,
which enable us to recast \eqnref{HFW Kutzelnigg} as
\begin{equation*}
H_\mathrm{FW} = mc^2+q\phi -[q\phi,X^\dag X]\frac{1}{2(1+X^\dag X)} +c\,\sdotp X.
\end{equation*}
Consequently, this leads to
\begin{widetext}
\begin{eqnarray}\label{HFW 2}
H_\mathrm{FW} &=& mc^2 + q\phi +c\sum_{j=0}^\infty a_j\frac{(-1)^j}{(2mc)^{2j+1}}\,(\sdotp)^{2j+2}
+q\hbar\sum_{j=1}^\infty b_j\frac{(-1)^j}{(2mc)^{2j}}\,\boldpi^{2j-2} (\Etimesp)\cdot\boldsigma\nonumber\\
&&\mbox{}+iq\hbar\sum_{j=1}^\infty (b_j+c_j)\frac{(-1)^j}{(2mc)^{2j}}\,\boldpi^{2j-2} (\Edotp)
-iq\hbar\left(\sum_{j=0}^\infty (j+1)a_{j+1}\frac{(-1)^j}{(2mc)^{2j+2}}\,
\boldpi^{2j}(\Edotp)\right)\frac{1}{1+X^\dag X}.
\end{eqnarray}
By \eqnref{combo identity a}--\eqnref{combo identity d}, it can be shown that the antihermitian (imaginary) parts in \eqnref{HFW 2} cancel each other out exactly as expected. Then, by \eqnref{series a} and \eqnref{series b}, it follows from \eqnref{HFW 2} that
\begin{equation}\label{HFW 2'}
H_\mathrm{FW} =
 mc^2+q\phi +mc^2\left(\sqrt{1+\left(\frac{\sdotp}{mc}\right)^2}-1\right)
+\frac{q\hbar}{2(mc)^2} \left(\frac{1}{1+\sqrt{1+\left(\frac{\boldpi}{mc}\right)^2}}
-\frac{1}{\sqrt{1+\left(\frac{\boldpi}{mc}\right)^2}}\right)
\boldsigma\cdot(\Etimesp).
\end{equation}
Up to the linear order in $\mathbf{B}$, we have
\begin{equation*}
\sqrt{1+\left(\frac{\sdotp}{mc}\right)^2}
=\sqrt{1+\left(\frac{\boldpi}{mc}\right)^2-\frac{q\hbar}{m^2c^3}\sdotB}
=\sqrt{1+\left(\frac{\boldpi}{mc}\right)^2} \left(1-\frac{1}{2}\frac{q\hbar}{m^2c^3} \frac{\sdotB}{1+\left(\frac{\boldpi}{mc}\right)^2}+\cdots\right).
\end{equation*}
Substituting this back to \eqnref{HFW 2'}, we obtain
\begin{equation}\label{HFW 4}
H_\mathrm{FW} = q\phi +\sqrt{m^2c^4+c^2\boldpi^2} -\frac{q\hbar}{2mc}\frac{1}{\gamma_{\boldpi}}\sdotB
+\frac{q\hbar}{2mc}\left(\frac{1}{\gamma_{\boldpi}}-\frac{1}{1+\gamma_{\boldpi}}\right) \boldsigma\cdot\left(\frac{\boldpi}{mc}\times\mathbf{E}\right),
\end{equation}
\end{widetext}
where $\gamma_{\boldpi}$ is defined as
\begin{equation}\label{gamma pi}
\gamma_{\boldpi}:=\sqrt{1+\left(\frac{\boldpi}{mc}\right)^2}
\equiv \sum_{n=0}^\infty {1/2 \choose n}\left(\frac{\boldpi}{mc}\right)^{2n}
\end{equation}
in accordance with the classical counterpart appearing in \eqnref{H spin}.
The FW transform of the Dirac Hamiltonian given in \eqnref{HFW 4} fully agrees with the classical counterpart \eqnref{H cl}--\eqnref{H spin} with $\mathbf{s}=\frac{\hbar}{2}\boldsigma$ and $\gamma'_m=0$ (or $\gamma_m=\frac{q}{mc}$).

Next, let us study the Dirac-Pauli theory. Consider the power series of $\calX$ in powers of $c^{-1}$:
\begin{equation*}
\calX:=X+X'=\sum_{j=1}^\infty\frac{\calX_j}{c^j}
=\sum_{j=1}^\infty\frac{X_j}{c^j}
+\sum_{j=1}^{\infty}\frac{X'_j}{c^j},
\end{equation*}
where $X$ and $X_j$ have been obtained. The constraint \eqnref{condition calX} together with \eqnref{recursion Xj} leads to
\begin{equation}\label{X'1 X'2 X'3}
X'_1=X'_2=0, \qquad 2mX'_3=-i\mu''\sdotE,
\end{equation}
and the recursion relations for $X'_n$ ($j\geq2$):
\begin{widetext}
\begin{subequations}\label{recursion X'j}
\begin{eqnarray}
\label{recursion a'}
2mX'_{2j}&=&q\left[\phi,X'_{2j-2}\right] +\mu''\left\{X_{2j-3}+X'_{2j-3},\sdotB\right\}
-\sum_{k_1+k_2=2j-1} \left(X_{k_1}\sdotp X'_{k_2}+X'_{k_1}\sdotp X_{k_2}+X'_{k_1}\sdotp X'_{k_2}\right)\nonumber\\
&&\mbox{}-i\mu''\sum_{k_1+k_2=2j-3}\left(X_{k_1}\sdotE\,X_{k_2}+X_{k_1}\sdotE\,X'_{k_2}+X'_{k_1}\sdotE\, X_{k_2}+X'_{k_1}\sdotp X'_{k_2}\right),\\
\label{recursion b'}
2mX'_{2j+1}&=& q\left[\phi,X'_{2j-1}\right] +\mu''\left\{X_{2j-2}+X'_{2j-2},\sdotB\right\}
-\sum_{k_1+k_2=2j} \left(X_{k_1}\sdotp X'_{k_2}+X'_{k_1}\sdotp X_{k_2}+X'_{k_1}\sdotp X'_{k_2}\right)\nonumber\\
&&\mbox{}-i\mu''\sum_{k_1+k_2=2j-2}\left(X_{k_1}\sdotE\,X_{k_2}+X_{k_1}\sdotE\,X'_{k_2}+X'_{k_1}\sdotE\, X_{k_2}+X'_{k_1}\sdotp X'_{k_2}\right).
\end{eqnarray}
\end{subequations}
\end{widetext}
Again, based on the high-order calculation conducted in \cite{Chen2013} up to $X'_{14}$, we have the following theorem:
\begin{theorem}\label{thm:2}
In the weak-field limit, we neglect nonlinear terms in $\mathbf{E}$ and $\mathbf{B}$. If the electromagnetic field is homogeneous (thus, $[\pi_i,E_j]=[\pi_i,B_j]=0$), the generic expression for $X'_{n\geq2}$ is given by
\begin{subequations}\label{X'2j and X'2j+1}
\begin{eqnarray}
\label{X'2j}
X'_{2j}&=&2b_{j-1}\frac{\mu''(-1)^{j}}{(2m)^{2j-2}}\,\boldpi^{2j-4}(\Bdotp),\\
\label{X'2j+1}
X'_{2j+1}&=&b_j\frac{i\mu''(-1)^{j}}{(2m)^{2j-1}}\,\boldpi^{2j-2}(\sdotE)\nonumber\\
&&\mbox{}
+d_j\frac{i\mu''(-1)^{j+1}}{(2m)^{2j-1}}\,\boldpi^{2j-4}(\sdotp)(\Edotp),\qquad
\end{eqnarray}
\end{subequations}
where we define $\mu'':=c\mu'$ for convenience
and the coefficients $d_j$ are defined as
\begin{subequations}\label{dj}
\begin{eqnarray}
d_{j=0}&=&d_{j=1}=0,\\
d_{j\geq2}&=&\sum_{j_1+j_2+j_3=j-2} 2(j_1+1)a_{j_1}a_{j_2}a_{j_3}.
\end{eqnarray}
\end{subequations}
\end{theorem}
\begin{proof}
The theorem is proven by applying \eqnref{recursion X'j} on \eqnref{X'1 X'2 X'3} inductively with the help of \eqnref{X2j and X2j+1}, \eqnref{combo identity a}, \eqnref{combo identity b}, and \eqnref{combo identity e}.
\end{proof}

Eq.~\eqnref{X'2j and X'2j+1} shows that $X'$ is of the order $\bigO(F_{\mu\nu})$. Consequently, up to $\bigO(F_{\mu\nu})$, \eqnref{calHFW Kutzelnigg} leads to
\begin{eqnarray*}\label{HFW+H'FW}
\calH_\mathrm{FW}
&=& mc^2 + \sqrt{1+X^\dag X}
\left(q\phi +c\,\sdotp X\right)
\frac{1}{\sqrt{1+X^\dag X}}\nonumber\\
&&\mbox{}
+\left(c\,\sdotp X' -\mu'\sdotB +i\mu'\sdotE\,X\right)\nonumber\\
&=:&H_\mathrm{FW}+H'_\mathrm{FW},
\end{eqnarray*}
where the first half part is identified as $H_\mathrm{FW}$ by \eqnref{HFW Kutzelnigg}, and the second half is called $H'_\mathrm{FW}$.
By \eqnref{X2j and X2j+1} and \eqnref{X'2j and X'2j+1}, we have
\begin{widetext}
\begin{eqnarray}\label{H'FW 2}
H'_\mathrm{FW}
&=& -2\mu'\sum_{j=1}^\infty b_j\frac{(-1)^j}{(2mc)^{2j}}\,\boldpi^{2j-2}(\sdotp)(\Bdotp)
+\mu'\left(\sum_{j=1}^\infty b_j\frac{(-1)^j}{(2mc)^{2j-1}}\,\boldpi^{2j-2}
-\sum_{j=0}^\infty a_j\frac{(-1)^j}{(2mc)^{2j+1}}\,\boldpi^{2j}\right)
(\Etimesp)\cdot\boldsigma\nonumber\\
&&\mbox{}-\mu'\sdotB
-i\mu'\sum_{j=0}^\infty \left(b_{j+1}-d_{j+1}+a_j\right) \frac{(-1)^j}{(2mc)^{2j+1}}\,\boldpi^{2j}(\Edotp),
\end{eqnarray}
where nonlinear terms in $F_{\mu\nu}$ have been neglected.
By \eqnref{combo identity f}, we find that the antihermitian part in \eqnref{H'FW 2} vanishes identically. Furthermore, by \eqnref{series a} and \eqnref{series b}, we have
\begin{eqnarray}\label{H'FW 3}
H'_\mathrm{FW}
&=& -\mu'\left(\frac{1}{1+\sqrt{1+\left(\frac{\boldpi}{mc}\right)^2}} -\frac{1}{\sqrt{1+\left(\frac{\boldpi}{mc}\right)^2}}\right)
\frac{(\sdotp)(\Bdotp)}{(mc)^2}
-\mu'\left(\frac{1}{\sqrt{1+\left(\frac{\boldpi}{mc}\right)^2}}\right)
\frac{(\Etimesp)\cdot\boldsigma}{mc}
-\mu'\sdotB\nonumber\\
&=&\mu'\left(\frac{1}{\gamma_{\boldpi}}-\frac{1}{1+\gamma_{\boldpi}}\right) \boldsigma\cdot\frac{\boldpi}{mc}\left(\frac{\boldpi}{mc}\cdot\mathbf{B}\right)
+\mu'\frac{1}{\gamma_{\boldpi}}\boldsigma\cdot\left(\frac{\boldpi}{mc}\times\mathbf{E}\right)
-\mu'\sdotB,
\end{eqnarray}
where $\gamma_{\boldpi}$ is defined in \eqnref{gamma pi}.
With \eqnref{HFW 4} and \eqnref{H'FW 3}, we have
\begin{eqnarray}\label{calHFW}
&&\calH_\mathrm{FW}(\mathbf{x},\mathbf{p},\boldsigma) = H_\mathrm{FW}+H'_\mathrm{FW}%\nonumber\\
=\sqrt{m^2c^4+c^2\boldpi^2}\,+q\phi(\mathbf{x})\nonumber\\
&&\qquad \mbox{}-\boldsigma\cdot\left[
\left(\mu'+\frac{q\hbar}{2mc}\frac{1}{\gamma_{\boldpi}}\right)\mathbf{B}
-\mu'\frac{1}{\gamma_{\boldpi}(1+\gamma_{\boldpi})}
\left(\frac{\boldpi}{mc}\cdot\mathbf{B}\right)\frac{\boldpi}{mc}
-\left(\mu'\frac{1}{\gamma_{\boldpi}} +\frac{q\hbar}{2mc}\frac{1}{\gamma_{\boldpi}(1+\gamma_{\boldpi})}\right)
\left(\frac{\boldpi}{mc}\times\mathbf{E}\right)
\right],
\end{eqnarray}
which is in complete agreement with the classical counterpart \eqnref{H cl}--\eqnref{H spin} with $\mathbf{s}=\frac{\hbar}{2}\boldsigma$ and $\mu'=\frac{\hbar}{2}\gamma'_m$.
\end{widetext}

This work was supported in part by the Ministry of Science and Technology of Taiwan under the Grants: No.\ 101-2112-M-002-027-MY3, No.\ 101-2112-M-003-002-MY3, and No.\ 101-2112-M-110-013-MY3.

\appendix

\section{Appendix: Useful formulae and lemmas}
The Pauli matrices satisfy the identity
$(\boldsigma\cdot\mathbf{a})(\boldsigma\cdot\mathbf{b}) =\mathbf{a}\cdot\mathbf{b}+i(\mathbf{a}\times\mathbf{b})\cdot\boldsigma$
for arbitrary vectors $\mathbf{a}$ and $\mathbf{b}$.
Meanwhile, we have
$(\boldsymbol{\nabla}\times\mathbf{a}+\mathbf{a}\times\boldsymbol{\nabla})\psi
=(\boldsymbol{\nabla}\times\mathbf{a})\psi$.
Consequently, we have
$(\sdotp)^2=\boldpi^2-\frac{q\hbar}{c}\,\sdotB$.

Neglecting any nonlinear terms in $\mathbf{E}$ and $\mathbf{B}$, by mathematical induction, we have
\begin{subequations}\label{identity 4}
\begin{eqnarray}
\label{identity 4a}
\left[\phi,\,(\sdotp)^{2n}\right]&=&(2n)i\hbar\,\boldpi^{2(n-1)}(\Edotp),\\
\label{identity 4b}
\left[\phi,\,(\sdotp)^{2n+1}\right] &=& i\hbar\,\boldpi^{2n}(\sdotE)\\ &&\mbox{}+(2n)i\hbar\,\boldpi^{2n-2}(\sdotp)(\Edotp).\nonumber
\end{eqnarray}
\end{subequations}

The coefficients defined in \eqnref{aj bj cj} and \eqnref{dj} give the Taylor series:
\begin{subequations}\label{series abcd}
\begin{eqnarray}
\label{series a}
\sum_{j=0}^\infty a_j\frac{(-1)^j}{2^{2j+1}}\,x^{2j+1}
&=&
\frac{x}{1\!+\!\sqrt{1\!+\!x^2}}
\equiv \frac{\sqrt{1\!+\!x^2}\!-\!1}{x},\\
\label{series b}
\sum_{j=1}^\infty b_j\frac{(-1)^j}{2^{2j}}\,x^{2j-2}
&=&
\frac{1}{2}\left(\frac{1}{1\!+\!\sqrt{1\!+\!x^2}}\!-\!\frac{1}{\sqrt{1\!+\!x^2}}\right),\\
%\equiv-\frac{1}{2\sqrt{1+x^2}}\left(\frac{1}{1+\sqrt{1+x^2}}\right),\\
\label{series c}
\sum_{j=2}^\infty c_j\frac{(-1)^j}{2^{2j}}\,x^{2j-4}
&=&
\frac{1}{8}\left(\frac{1}{1\!+\!\sqrt{1\!+\!x^2}}\!-\!\frac{1}{\sqrt{1\!+\!x^2}}\right)^2\!,\\
%\equiv\frac{1}{8(1+x^2)}\left(\frac{1}{1+\sqrt{1+x^2}}\right)^2,\qquad
\label{series d}
\sum_{j=2}^\infty d_j\frac{(-1)^{j}}{2^{2j-1}}\,x^{2j-4}
&=&
\frac{1}{\sqrt{1+x^2}} \left(\frac{1}{1+\sqrt{1+x^2}}\right)^2\!,\qquad\;
\end{eqnarray}
\end{subequations}
which lead to the combinatorial identities for $j\geq1$:
%\begin{subequations}\label{combo identities}
\begin{eqnarray}
\label{combo identity a}
%\text{for}\ j\geq1:\qquad
\sum_{j_1+j_2=j-1}a_{j_1}a_{j_2}&=&a_j,\\
\label{combo identity b}
2\sum_{j_1+j_2=j-1}a_{j_1}b_{j_2}&=&b_j-a_{j-1}\equiv2(j-1)a_{j-1},\quad\\
\label{combo identity c}
2\sum_{j_1+j_2=j-1}a_{j_1}c_{j_2}&\equiv&4\sum_{j_1+j_2+j_3=j-1}a_{j_1}b_{j_2}b_{j_3}\nonumber\\
=c_j-b_j+a_j&\equiv& c_j-2(j-1)a_{j-1},
\end{eqnarray}
%\end{subequations}
for $j\geq0$:
\begin{eqnarray}
\label{combo identity d}
&&b_{j+1}+c_{j+1}=4b_j+4c_j+a_j,\\
\label{combo identity f}
&&b_{j+1}+a_j=d_{j+1},
\end{eqnarray}
and for $j\geq2$:
\begin{equation}\label{combo identity e}
2\sum_{j_1+j_2=j-1} a_{j_1}d_{j_2} + 2a_{j-1} = d_j.
\end{equation}


\begin{thebibliography}{99}

\bibitem{Dirac1928}
  P.~A.~M.~Dirac,
  ``The Quantum Theory of the Electron,''
  Proc.\ R.\ Soc.\ London {\bf 117}, 610 (1928).

\bibitem{Dirac1982} P.~A.~M.~Dirac, {\it Principles of Quantum Mechanics}, 4th ed.\ (Clarendon, Oxford, 1982).

\bibitem{Pauli1941}
  W.~Pauli,
  ``Relativistic Field Theories of Elementary Particles,''
  Rev.\ Mod.\ Phys.\ {\bf 13}, 203 (1941).

\bibitem{Foldy1950}
  L.~L.~Foldy and S.~A.~Wouthuysen,
  ``On the Dirac theory of spin 1/2 particle and its nonrelativistic limit,''
  Phys.\ Rev.\  {\bf 78}, 29 (1950).

\bibitem{Strange2008} P.~Strange, {\it Relativistic Quantum Mechanics}, 1st ed.\ (Cambridge University Press, Cambridge, United Kingdom, 2008).

\bibitem{Silenko1995-analysis}
  A.~J.~Silenko,
  ``Comparative analysis of direct and `step-by-step' Foldy-Wouthuysen transformation methods,''
  Theor.\ Math.\ Phys.\  {\bf 176},  987 (2013).

\bibitem{Reiher2009 Book}
   M.~Reiher and A.~Wolf, {\it Relativistic Quantum Chemistry}, (WILEY-VCH Verlag GmbH \& Co.\ KGaA, Weinheim, 2009).

\bibitem{Thomas1927}
  L.~H.~Thomas,
  ``The kinematics of an electron with an axis,''
  Phil.\ Mag.\ Ser.\ 7 {\bf 3}, 1 (1927).

\bibitem{Bargmann1959}
  V.~Bargmann, L.~Michel and V.~L.~Telegdi,
  ``Precession of the polarization of particles moving in a homogeneous electromagnetic field,''
  Phys.\ Rev.\ Lett.\  {\bf 2}, 435 (1959).

\bibitem{Jackson1999}
  J.~D.~Jackson,
  {\it Classical Electrodynamics}, 3rd ed.\
  (John Wiley \& Sons, New York, 1999); Chapter 11.

\bibitem{Chen2014}
 T.-W.~Chen and D.-W.~Chiou,
 ``Correspondence between classical and Dirac-Pauli spinors in view of the Foldy-Wouthuysen transformation,''
 Phys.\ Rev.\ A {\bf 89}, 032111 (2014) [arXiv:1310.8513 [quant-ph]].

\bibitem{Silenko2008}
 A.~J.~Silenko,
  ``Foldy-Wouthyusen Transformation and Semiclassical Limit for Relativistic Particles in Strong External Fields,''
  Phys.\ Rev.\ A {\bf 77}, 012116 (2008)
  [arXiv:0710.4218 [math-ph]].

\bibitem{Chen2010}
  T.-W.~Chen and D.-W.~Chiou,
  ``Foldy-Wouthuysen transformation for a Dirac-Pauli dyon and the Thomas-Bargmann-Michel-Telegdi equation,''
  Phys.\ Rev.\ A {\bf 82}, 012115 (2010)
  [arXiv:1005.4128 [quant-ph]].

\bibitem{Chen2013}
  T.-W.~Chen and D.-W.~Chiou,
  ``High-order Foldy-Wouthuysen transformations of the Dirac and Dirac-Pauli Hamiltonians in the weak-field limit,''
  Phys.\ Rev.\ A {\bf 90}, 012112 (2014) [arXiv:1311.3432 [quant-ph]].

\bibitem{Rubinow1963}
  S.~I.~Rubinow and J.~B.~Keller,
  ``Asymptotic Solution of the Dirac Equation,''
  Phys.\ Rev.\ {\bf 131}, 2789 (1963).

\bibitem{Rafa1964}
  K.~Rafanelli and R.~Schiller,
  ``Classical Motions of Spin-1/2 Particles,''
  Phys.\ Rev.\ {\bf 135}, B279 (1964).

\bibitem{Froh1993}
  J.~Fr\"{o}hlich and U.~M.~Studer,
  ``Gauge invariance and current algebra in nonrelativistic many body theory,''
  Rev.\ Mod.\ Phys.\  {\bf 65}, 733 (1993).

\bibitem{Silenko1995}
  A.~Y.~Silenko,
  ``Dirac equation in the Foldy-Wouthuysen representation describing the interaction of spin 1/2 relativistic particles with an external electromagnetic field,''
  Theor.\ Math.\ Phys.\  {\bf 105}, 1224 (1995)
  [Teor.\ Mat.\ Fiz.\  {\bf 105}, 46 (1995)].

\bibitem{Chen2014b}
  D.~W.~Chiou and T.~W.~Chen,
  ``Exact Foldy-Wouthuysen transformation of the Dirac-Pauli Hamiltonian in the weak-field limit by the method of direct perturbation theory,''
  Phys.\ Rev.\ A {\bf 94}, no. 5, 052116 (2016)
  %doi:10.1103/PhysRevA.94.052116
  [arXiv:1405.4495 [quant-ph]].

\bibitem{Rutkowski1986a}
 A.~Rutkowski,
 ``Relativistic perturbation theory: I.\ A new perturbation approach to the Dirac equation,''
 J.\ Phys.\ B: At.\ Mol.\ Phys.\ {\bf 19} 149 (1986).

\bibitem{Rutkowski1986b}
 A.~Rutkowski,
 ``Relativistic perturbation theory: II.\ One-electron variational perturbation calculations,''
 J.\ Phys.\ B: At.\ Mol.\ Phys.\ {\bf 19} 3431 (1986).

\bibitem{Rutkowski1986c}
 A.~Rutkowski,
 ``Relativistic perturbation theory: III.\ A new perturbation approach to the two-electron Dirac-Columb equation,''
 J.\ Phys.\ B: At.\ Mol.\ Phys.\ {\bf 19} 3443 (1986).

\bibitem{Heully1986}
  J.~L. Heully, I.~Lindgren, E.~Lindroth, S.~Lundgvist, and A.-M.~M{\aa}rtensson-Pendrill,
  ``Diagonalisation of the Dirac Hamiltonian as a basis for a relativistic many-body procedure,''
  J.\ Phys.\ B: At.\ Mol.\ Phys.\ {\bf 19} 2799 (1986).

\bibitem{Kutzelnigg1990}
  W.~Kutzelnigg,
  ``Perturbation theory of relativistic corrections,''
  Z.\ Phys.\ D {\bf 15}, 27 (1990).


\end{thebibliography}
\end{document}